\documentclass[conference]{IEEEtran}
\IEEEoverridecommandlockouts
\usepackage{nopageno}
\usepackage[T1]{fontenc}
\usepackage[utf8]{inputenc}
\usepackage[noadjust]{cite}
\usepackage{amsmath,amssymb,amsfonts,amsthm}
\usepackage{mathtools}
\usepackage{algorithmic}
\usepackage[ruled]{algorithm2e}
\usepackage{listings}
\usepackage{alltt}
\usepackage{graphicx}
\usepackage{lipsum}
\usepackage{comment}
\usepackage[warn]{textcomp}
\usepackage{xcolor}
\usepackage{url}
\usepackage{xspace}
\usepackage[final]{microtype}
\usepackage{alltt}
\usepackage{tcolorbox}
\usepackage{hyperref}
\usepackage{color}
\usepackage{threeparttable}
\usepackage{mdframed}
\usepackage[figuresright]{rotating}
\usepackage{makecell}
\usepackage{color,soul}
\usepackage{tikz}
\usetikzlibrary{shapes.symbols}
\usetikzlibrary{shapes.arrows}
\usetikzlibrary{calc}
\usetikzlibrary{positioning}
\usetikzlibrary{intersections}

\usepackage{multirow}

\usepackage{booktabs}

\renewcommand{\thetable}{\arabic{table}}

\newtheorem{definition}{Definition}
\newtheorem{lemma}{Lemma}

\usepackage{subcaption}
\DeclareCaptionSubType*[Alph]{table}
\DeclareCaptionLabelFormat{mystyle}{Table~\bothIfFirst{#1}{ }#2}
\newcommand{\solver}[1]{\texttt{#1}\xspace}

\begin{document}

\newcommand{\placeholder}[1]{\rule{1pt}{#1}}

\newboolean{shownotes}
\setboolean{shownotes}{false}
\newboolean{showappendix}
\setboolean{showappendix}{true}

\ifshownotes
\newcommand{\fl}[1]{\begin{mdframed}{{\textcolor{red}{Florian: #1}}}\end{mdframed}}
\newcommand{\TODO}[1]{{$\blacksquare$\textcolor{red}{\textbf{TODO}}(#1)}}
\newcommand{\mnote}[1]{\textbf{NOTE}(#1)\marginpar{$\blacksquare$}}
\newcommand{\marginnote}[1]{\marginpar{$\blacksquare$ #1}}
\else
\newcommand{\fl}[1]{}
\newcommand{\TODO}[1]{}
\newcommand{\mnote}[1]{}
\newcommand{\marginnote}[1]{}
\fi

\title{\stitch: Proof Combination for Divide-and-Conquer SAT Solvers\\

\thanks{\textbf{This article will appear in the proceedings of Formal Methods in
Computer-Aided Design (FMCAD 2022).}}
}

\newcommand{\xspacemm}{\ifmmode\else\xspace\fi}
\newcommand{\stitch}{\textit{Proof-Stitch}\xspacemm}
\newcommand{\CL}{\ensuremath{\mathit{CL}_{\mathit{avg}}}}

\author{
  \IEEEauthorblockN{\large Abhishek Nair, Saranyu Chattopadhyay, Haoze Wu, Alex Ozdemir, and Clark Barrett }
  \IEEEauthorblockA{\large Stanford University, Stanford, USA.}
  \IEEEauthorblockA{\large \{aanair, saranyuc, haozewu, aozdemir, barrettc\}@stanford.edu}
}

\maketitle
\begin{abstract}  

With the increasing availability of parallel computing power, there is a growing focus on parallelizing algorithms for important automated reasoning problems such as Boolean satisfiability (SAT).
Divide-and-Conquer (\textit{D\&C}) is a popular parallel SAT solving paradigm that partitions SAT instances into independent sub-problems which are then solved in parallel.
For unsatisfiable instances, state-of-the-art D\&C solvers generate DRAT refutations for each sub-problem. However, they do not generate a single refutation for the original instance. To close this gap,  
we present \stitch, a procedure 
for combining refutations of different sub-problems into a single refutation for the original instance. We prove the correctness of the procedure and propose optimizations to reduce the size and checking time of the combined refutations by invoking existing trimming tools in the proof-combination process. We also provide an extensible implementation of the proposed technique.
Experiments on instances from last year's SAT competition show that the optimized refutations are checkable up to seven times faster than unoptimized refutations.
\end{abstract}


\section{Introduction}

Boolean satisfiability (SAT) solvers
have improved dramatically in recent years.
They are now regularly
used in a wide variety of application areas
including hardware verification~\cite{hardware},
computational biology~\cite{bio} and decision planning~\cite{planning}.

With the emergence of cloud-computing and improvements in multi-processing hardware,
the availability of parallel computing power has also increased dramatically.
This has naturally led to an increased focus on
parallelizing important algorithms, and SAT is no exception.
There are two traditional approaches to parallel SAT solving - the Divide-and-Conquer (\textit{D\&C})
approach~\cite{blochinger2003parallel,hyvarinen2010partitioning,hyvarinen2006distribution}
and the portfolio approach~\cite{xu2008satzilla}.
In the D\&C approach, the
original SAT instance is partitioned into independent sub-problems to be solved
in parallel,
while in the portfolio approach multiple SAT solvers are
independently run on the original instance. 
Although the portfolio approach in combination with clause sharing performs well for small portfolio sizes, the D\&C approach scales better in
environments with large parallel computing power such as the cloud.
Several
implementations of D\&C solvers exist~\cite{blochinger2003parallel,hyvarinen2010partitioning,hyvarinen2006distribution,ggsatOWC}.
Every implementation uses:
a \textit{divider} to split up the original instance into sub-problems,
and a \textit{base SAT solver} to solve the independent sub-problems. 
For example, \texttt{ggSAT}~\cite{ggsatOWC} uses CadiCaL~\cite{cadical} as its base
solver. 

If a SAT problem is unsatisfiable,
a proof of unsatisfiability (or \textit{refutation})
can be produced and independently checked
to validate the result.
Since 2013, the annual SAT competition has required
SAT solvers
to generate refutations.
The most commonly supported refutation format today is the DRAT format~\cite{heule2015proofs}.
Existing
D\&C SAT solvers produce refutations for each
sub-problem independently.
%
%
However, even if the refutation for each sub-problem passes the proof-checker, this is not a formal guarantee that the original instance also admits a refutation, as there could have been an error in the partitioning strategy. 
For example, a buggy solver may incompletely partition the SAT instance $(\neg \ell_1) \wedge (\ell_2 \vee \ell_3) \wedge (\neg \ell_2 \vee \ell_3)$ into sub-problems with cubes $\ell_1$ and $\lnot \ell_2$.  Both of these sub-problems are unsatisfiable, even though the instance is satisfiable. Transient errors in the underlying distributed system may also cause sub-problem refutations to be truncated or missing.
To address these challenges, we introduce \stitch, which implements a strategy for combining DRAT refutations for sub-problems into a single refutation for the original instance, a process we call \emph{refutation stitching}.
%
Our contributions are:
\begin{itemize}
  
  \item We describe an algorithm for combining DRAT refutations of partitions of problems into a single refutation for the original problem and provide an open-source implementation  on GitHub~\cite{github}.

  \item We describe an optimization technique leveraging existing trimming tools (e.g., \emph{drat-trim}~\cite{wetzler2014drat}) to improve the quality of the combined refutations.
  
  \item We evaluate our implementation on benchmarks from last year's SAT competition~\cite{sat2021}. Our results show that trimmed refutations are checkable up to seven times faster than untrimmed refutations.

\end{itemize}

The rest of this paper is organized as follows. 
Section \ref{sec:related_work}
discusses background and related work.
Section \ref{sec:methodology} presents the \stitch
algorithm and theoretically justifies our method of combining refutations.
We also describe an optimization technique that reduces the checking time and the size of the combined refutations.
Section \ref{sec:implementation} details our tool implementation.
Results are presented
in Section \ref{sec:results}, and Section \ref{sec:conc} concludes.

\section{Background and Related Work}
\label{sec:related_work}

\subsection{Propositional refutations}
We assume familiarity with the basic concepts of CDCL SAT algorithms (see, e.g.,
\cite{sathandbook}).
We also assume that a base SAT solver can produce a \emph{DRAT} refutation, which we define below (following~\cite{heule2010clause}).

Throughout the paper
we model clauses as \textit{sets} of literals
and formulas as \textit{multisets} of clauses.
By $\cdot \cup \cdot$,
we denote the standard union operation on sets,
and the multiplicity-summing union on multisets.
%

%
Let $F = \{C_1, \dots, C_n\}$ be a formula.
$F$ \textit{unit propagates on $\ell$} to $F' = \{C\setminus\{\lnot\ell\} : C \in
F, \ell\not\in C \} \cup \{\ell\}$ (written $F \to_\ell F'$)
if there exists a clause
$\{\ell, \ell_1, \dots, \ell_k\} \in F$
such that $\{\lnot\ell_i\} \in F$ for
$i \in [1,k]$.
If $F \to_\ell F'$ for some $\ell$, then $F \to F'$.
We say that $F\to\bot$ if $F$ contains an empty clause.
Let the relation $\to^*$ denote the reflexive, transitive closure of $\to$.
We say that $F\mapsto F'$
when $F\to^*F'$ and there is no $F'' \ne F'$ such that $F'\to F''$.
One can show that the $\mapsto$ relation is a function.
We say that $C = \{\ell_1, \dots, \ell_k\}$
has \textit{asymmetric tautology} (AT)
with respect to $F$ if $F \cup \{\lnot \ell_1\} \cup \dots \cup \{\lnot \ell_k\}
\mapsto \bot$.
We say that $C$
has \textit{resolution asymmetric tautology} (RAT)
with respect to literal $\ell_1 \in C$ and $F$
if for all $C' \in F$ containing $\lnot \ell_1$,
$C \cup (C' \setminus \{\lnot \ell_1\})$
has AT.
%

Let $o_i$ denote an operation.
Consider a sequence of operation-clause pairs $\pi = ((o_1, C_1), \dots,
(o_m, C_m))$, where 
each $o_i$ indicates either the addition ($\oplus$) or deletion ($\ominus$) of a clause from a formula.

Let $\phi$ denote a CNF formula.
Define
$\phi_i$ recursively: $\phi_0 = \phi$,
and
$\phi_{i+1}$ is $\phi_i \cup \{C_{i+1}\}$ when $o_{i+1}$ is $\oplus$,
or $\phi_i \setminus \{C_{i+1}\}$ otherwise.
The sequence $\pi$
is a \textit{DRAT refutation}
of $\phi$
if when $o_{i+1}=\oplus$ then $C_{i+1}$
has RAT with respect to
$\phi_{i}$,
and if the last element in $\pi$ is $(\oplus, \emptyset)$.

\subsection{Divide-and-Conquer SAT solving}
One parallel SAT solving paradigm is \textit{Divide-and-Conquer}:
a SAT instance is divided into simpler SAT instances (sub-problems),
which are then solved in parallel. Typically, the sub-problems represent partitions of the search space,
such that the disjunction of all the sub-problems is equisatisfiable with the
original problem. 
The sub-problems are derived from the original instance by assigning Boolean values to literals.
The set of literals that are assigned (decided) for a particular sub-problem is called
the \textit{cube} of the sub-problem and the number of literals in the cube is the \textit{depth} of the sub-problem.
There are many D\&C-based
solvers~\cite{blochinger2003parallel,hyvarinen2010partitioning,hyvarinen2006distribution},
including:
\solver{Psato}~\cite{zhang1996psato},
\solver{Painless}~\cite{le2019modular},
and \solver{AMPHAROS}~\cite{nejati2017propagation}.
One prominent D\&C approach, Cube-and-Conquer~\cite{DBLP:conf/hvc/HeuleKWB11},
uses a lookahead solver to divide instances and a CDCL solver to solve
sub-problems.
This approach has been
successful for large mathematical problems~\cite{DBLP:conf/sat/HeuleKM16}
and is implemented by tools such as
\solver{Paracooba}~\cite{heisinger2020parac}
and \solver{gg-sat}~\cite{ggsatOWC}.

D\&C SAT solvers generate separate DRAT refutations for each sub-problem.
There has been little work
on combining these refutations
into a single refutation for the original instance.
One work~\cite{heule2015compositional}
considers proof composition,
but its parallel composition rule does not apply to DRAT refutations.
Another work~\cite{philipp2016unsatisfiability} gives an alternate proof
calculus for parallel solvers.

\section{Methodology}
\label{sec:methodology}
In this section,
we present an algorithm to combine sub-problem refutations into a refutation for the original Boolean instance.
Then we show the algorithm's correctness.
Finally, we present a technique to optimize the combined refutations.

\subsection{Algorithm}
The first step in the \stitch algorithm is to construct a decision tree representing the steps taken by the D\&C solver. The root of the tree represents the original instance, and the leaves represent the sub-problems. Figure \ref{fig:decision} shows the decision tree for an example instance.

\begin{algorithm}
\caption{Stitching algorithm}
\label{alg:stitch}
  \SetKwInOut{Input}{In}
  \SetKwInOut{Output}{Out}
  \Input{%
    Instance: $\phi$,\\
    Decision literal: $x$,\\
    Refutations of:\\
    $\phi \cup \{\{x\}\}$: $\;\; \pi = ((o_1, C_1), \dots, (o_n, C_{n}))$,\\
    $\phi \cup \{\{\lnot x\}\}$: $\pi' = ((o'_1, C'_1), \dots, (o'_m, C'_{m}))$,\\
  }
  \Output{Refutation of $\phi$}
  \textbf{procedure} \emph{stitching} $(\phi, x, \pi, \pi')$
  
  \Return 
\begin{align*}
  \Big(\;(o_1, C_1 \cup \{\lnot x\}),\; \dots,\; (o_n, C_n\cup\{\lnot x\}),\\
  (o'_1, C'_1 \cup \{x\}),\; \dots,\; (o'_m, C'_m\cup\{x\}) &,\;
  (\oplus, \emptyset)\;\Big)
\end{align*}
\end{algorithm}

Next, \stitch performs a sequence of \emph{stitching} operations to produce a single refutation for the original SAT instance. A stitching operation (Algorithm \ref{alg:stitch}) reads in a SAT instance $\phi$, a decision variable $x$ and two refutations $\pi$ and $\pi'$ corresponding to the sub-problems $\phi \cup \{\{x\}\}$ and $\phi \cup \{\{\lnot x\}\}$ respectively.
It produces a single refutation corresponding to the instance $\phi$. The refutation for instance $\phi$ contains the clauses
from refutation $\pi$ appended with the literal $\lnot x$ and the clauses from refutation $\pi'$ appended with the literal $x$.  More generally, the clauses from a refutation are appended with the negation of the decision literal used to generate the sub-problem.
Figure \ref{fig:example} illustrates the stitching operation.

As an example of the proof combination process, consider Figure \ref{fig:stitching}. First the refutations $\pi_{00}$ and $\pi_{01}$ are combined. Then $\pi_{10}$ and $\pi_{11}$ are combined, and finally, $\pi_{0}$ and $\pi_{1}$ are combined to produce the refutation $\pi$ corresponding to the original instance. 
In \stitch, the stitching operations are ordered according to the following rule: A stitching operation to combine a pair of refutations $\pi$ and $\pi'$ can only occur after all refutations with greater depth have been combined. 
Informally, this means that refutations are combined in decreasing order of their depth, as shown in Figure \ref{fig:stitching}. Stitching operations at the same depth are independent and can occur in parallel.

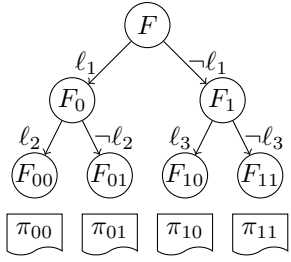
\begin{figure}[t]
    \centering
    \begin{tikzpicture}
  \begin{scope}
    [
     level distance = 10mm,
     level 1/.style={sibling distance = 20mm},
     level 2/.style={sibling distance = 10mm},
     ->,
     inner sep = 0mm,
     minimum width = 6mm,
    ]
    \path[]
      node[circle,draw] {$F$}
      child {
        node[circle,draw] {$F_{0}$}
        child {
          node[circle,draw] (F00) {$F_{00}$}
          edge from parent node[left] {$\ell_2$}
        }
        child {
          node[circle,draw] (F01) {$F_{01}$}
          edge from parent node[right] {$\lnot\ell_2$}
        }
        edge from parent node[left] {$\ell_1$}
      }
      child {
        node[circle,draw] {$F_{1}$}
        child {
          node[circle,draw] (F10) {$F_{10}$}
          edge from parent node[left] {$\ell_3$}
        }
        child {
          node[circle,draw] (F11) {$F_{11}$}
          edge from parent node[right] {$\lnot\ell_3$}
        }
        edge from parent node[right] {$\lnot\ell_1$}
      }
    ;
  \end{scope}
  \begin{scope}
    [
     tape,
     tape bend top = none,
     tape bend height = 3pt
    ]
    \path
      (F00.south) node[draw, below = 2mm] {$\pi_{00}$}
      (F01.south) node[draw, below = 2mm] {$\pi_{01}$}
      (F10.south) node[draw, below = 2mm] {$\pi_{10}$}
      (F11.south) node[draw, below = 2mm] {$\pi_{11}$}
      ;
  \end{scope}
\end{tikzpicture}
    \caption{Decision tree of an example unsatisfiable SAT instance.}
    \label{fig:decision}
\end{figure}

\begin{figure}[t]
    \centering
    \begin{tikzpicture}
  [tape bend top = none,
   tape bend height = 6pt,
   shorten >=2mm,
   shorten <=2mm,
   refutation/.style={
     minimum width = 7mm,
     tape,
     draw,
     below = 2mm,
     align = center
   },
   ->
  ]
  \path
    node[refutation, align=center]
    (pf1)
    {
      $\{\ell_1, \ell_2, \ell_3\}$
      \\
      $\{\ell_2, \ell_5\}$
      \\
      $\{\ell_4, \ell_5\}$
      \\
      $\{\}$
    }
    (2.9,0)
    node[refutation, align=center]
    (pf2)
    {
      $\{\ell_1, \ell_2, \ell_3, \lnot\ell_7\}$
      \\
      $\{\ell_2, \ell_5, \lnot\ell_7\}$
      \\
      $\{\ell_4, \ell_5, \lnot\ell_7\}$
      \\
      $\{\lnot\ell_7\}$
    }
    (0,-2.5)
    node[refutation, align=center]
    (pf3)
    {
      $\{\ell_4, \ell_2\}$
      \\
      $\{\ell_3, \ell_5\}$
      \\
      $\{\}$
    }
    (2.9,-2.5)
    node[refutation, align=center]
    (pf4)
    {
      $\{\ell_4, \ell_2, \ell_7\}$
      \\
      $\{\ell_3, \ell_5, \ell_7\}$
      \\
      $\{\ell_7\}$
    }
    (6.0,0)
    node[refutation, align=center]
    (pf5)
    {
      $\{\ell_1, \ell_2, \ell_3, \lnot\ell_7\}$
      \\
      $\{\ell_2, \ell_5, \lnot\ell_7\}$
      \\
      $\{\ell_4, \ell_5, \lnot\ell_7\}$
      \\
      $\{\lnot\ell_7\}$
      \\
      $\{\ell_4, \ell_2, \ell_7\}$
      \\
      $\{\ell_3, \ell_5, \ell_7\}$
      \\
      $\{\ell_7\}$
      \\
      $\{\}$
    }
    (pf1.east) edge[->, line width = 2pt] node[above] {$\lnot\ell_7$} (pf2.west)
    (pf3.east) edge[->, line width = 2pt] node[above] {$\ell_7$} (pf4.west)
    (pf2.east) edge[->, line width = 2pt] node[above] {} (pf5.west)
    (pf4.east) edge[->, line width = 2pt] node[above] {} (pf5.west)
    ;
\end{tikzpicture}
    \caption{Stitching operation on example refutations}
    \label{fig:example}
\end{figure}
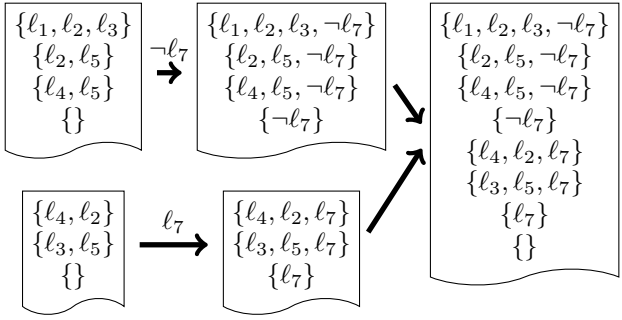

\begin{figure*}[t]
    \centering
    \includegraphics[width=0.8\textwidth]{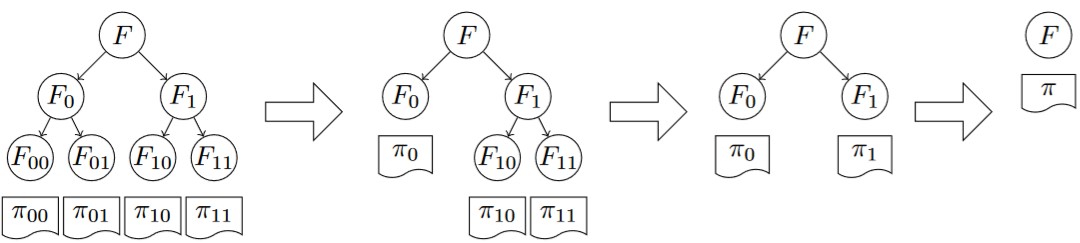}
    \caption{Refutation stitching process for the SAT instance shown in Figure \ref{fig:decision}. The decision literals are omitted.}
    \label{fig:stitching}
\end{figure*}

\subsection{Justification for the stitching operation}
\label{sec:justify}

We now show that Algorithm~\ref{alg:stitch} is correct:
given suitable inputs,
it produces a DRAT refutation for $\phi$.

\begin{definition}
  A DRAT refutation $\pi$ is \textbf{preserving} if
  for all $C$, $(\ominus, C)$ occurs at most as many times in $\pi$
  as $(\oplus, C)$.
  \label{defn:preserving}
\end{definition}

\begin{lemma}
    Let $\phi$ be a CNF formula,
    $x$ be a variable,
    and $\pi$ and $\pi'$
    be preserving DRAT refutations
    of $\phi \cup \{\{x\}\}$
    and $\phi \cup \{\{\lnot x\}\}$ respectively.
    Then, stitching$(\phi, x, \pi, \pi')$
    outputs a preserving DRAT refutation of $\phi$.
    \label{lemma:cor}
\end{lemma}

\begin{proof}
Let $\pi^*$ be the output of \emph{stitching}.
Let
$\pi = ((o_1, C_1), \dots, (o_n, C_n))$
and
$\pi' = ((o'_1, C'_1), \dots, (o'_{n'}, C'_{n'}))$.
Let $\psi = \phi \cup \{\{x\}\}$
and
$\psi' = \phi \cup \{\{\lnot x\}\}$.
Define $\psi_i$ recursively, by
$\psi_0 = \psi$
and $\psi_{i+1} = \psi_i \cup \{C_{i+1}\}$
when $o_{i+1}$ is an addition,
and $\psi_{i+1} = \psi_i \setminus \{C_{i+1}\}$ otherwise.
Define $\psi'_i$ (respectively $\phi_i$) analogously,
based on formula $\psi'$ (resp. $\phi$) and refutation $\pi'$ (resp. $\pi^*$).

By construction, $\pi^*$'s final step is $(\oplus, \emptyset)$.
Moreover, since $\pi$ and $\pi'$ are preserving and formulas are clause
\textit{multisets}, $\pi^*$ is preserving.
Thus, our main task is to show that
each addition $(\oplus, C^*_{i+1})$ in $\pi^*$
has RAT with respect to $\phi_i$.
$C^*_{i+1}$ is either derived from a clause in $\pi$,
derived from a clause in $\pi'$,
or is the final empty clause.
We begin with the first case:
$C^*_{i+1} = C_{j+1} \cup \{\lnot x\}$.

First, we show that if $C_{j+1}$ has AT with respect to
$\psi_j$, then $C^*_{i+1}$ has AT with respect
to $\phi_i$.
Note that $\psi_j \cup \{\{\lnot \ell_1\}, \dots, \{\lnot \ell_k\}\}
= F'\cup\{\{x\}\} \cup\{\{\lnot \ell_1\}, \dots, \{\lnot \ell_k\}\}
\to_x F''\cup\{\{x\}\} \cup\{\{\lnot \ell_1\}, \dots, \{\lnot \ell_k\}\} \mapsto\bot $.
Now, consider $F''' = \phi_i \cup \{\{x\}, \{\lnot \ell_1\}, \dots, \{\lnot \ell_k\}\}$.
If $F''' \mapsto \bot$, then $C^*_{i+1}$ has the desired property.
Observe that $F''' \to_x F''\cup\{\{x\}\} \cup\{\{\ell_1\}, \dots, \{\ell_k\}\}$;
thus, since the latter propagates to bottom, $F'''$ does too.

Second, we show that if $C_{j+1}$ has RAT with respect to
literal $\ell$ and formula $\psi_j$,
then $C^*_{i+1} = \{\lnot x\}\cup C_{j+1}$ has RAT with respect to literal $\ell$ and formula $\phi_i$.
Let $C^*$ be a clause in $\phi_i$ that contains $\lnot \ell$.
If $C^*_{i+1} \cup (C^* \setminus \{\lnot \ell\})$ has AT
with respect to $\phi_i$, we are done.
Since $C^*$ is a clause in $\phi_i$,
there is some $C$ in $\psi_j$
such that $C \cup \{\lnot x\} = C^*$ or $C=C^*$.
Thus, 
$C^*_{i+1} \cup (C^* \setminus \{\lnot \ell\}) =
\{\lnot x\} \cup C_{j+1} \cup (C \setminus \{\lnot\ell\})$.
Let $\lnot x, \ell_1, \dots, \ell_k$ be the literals of this clause.
As before, since
$\psi_j \cup \{\{\lnot \ell_1\}, \dots, \{\lnot \ell_k\}\}$ unit propagates to bottom,
$\phi_i \cup \{\{x\}, \{\lnot \ell_1\}, \dots, \{\lnot \ell_k\}\}$ does too.

In the case
that
$C^*_{i+1} = C'_{j+1} \cup \{x\}$
(i.e., $C^*_{i+1}$ is derived from $\pi'$),
the argument is similar.
The key insight is that an initial propagation on $\lnot x$ in any AT check
removes all the clauses added by $\pi$.
Since $\pi$ deletes no clauses from the original formula,
this leaves an intermediate propagation result that shows $C'_{j+1}$ is RAT.

The final step in $\pi^*$ is $(\oplus, \emptyset)$.
It has AT because $\phi_{n+m}$ contains both $\{x\}$ and $\{\lnot x\}$.
Since $\pi^*$'s added clauses all have the AT or RAT properties, and the final step adds an empty clause,
$\pi^*$ is a valid DRAT refutation of $\phi$.

\end{proof}

 In \stitch, the final refutation is built through stitching operations on DRAT refutations of the sub-problems. Since each stitching operation produces a preserving DRAT refutation, recursive application of Lemma \ref{lemma:cor} proves that the final refutation is a valid DRAT refutation of the original instance.

\subsection{Optimization}

Empirically, we have observed that refutations created through stitching operations contain a large number of clauses that are not needed during validation ("redundant" clauses).
Identifying and removing these clauses reduces the time required to check the refutation and the storage space required to save the refutation.
One approach to remove such redundant clauses is by identifying the "unsatisfiable core" as described in~\cite{goldberg2003verification}.
This approach optimizes the refutation by only retaining clauses that are essential for validation by a proof-checker.
Our implementation optimizes refutations by using \emph{drat-trim} to extract the unsatisfiable core after every stitching operation.

However, aggressively invoking the optimization technique (e.g., after every stitching operation) could incur significant run-time overhead in the refutation generation process. This calls for a heuristic to decide when to apply the optimization technique.
%
%
%
Empirically we observe that refutations with larger clauses (more literals) require longer to check.
We hypothesize that this occurs because larger clauses are less likely to contribute to unit-propagation while
simultaneously consuming more memory in the cache of the refutation checker.
Therefore, optimizing refutations with large clauses should yield the greatest benefit.
To implement this, we introduce a threshold parameter \CL.
After each stitching step,
the refutation is optimized only if the average clause length in the refutation is greater than \CL.
%

\section{Implementation}
\label{sec:implementation}
In this section, we describe our implementation of the \stitch algorithm. \stitch is implemented in Python and uses \emph{drat-trim}~\cite{wetzler2014drat} to optimize refutations. 
Our tool comprises of just under 300 lines of Python code and is available on GitHub~\cite{github}.

The tool inputs are the original SAT instance in CNF form, the refutations and cubes for each sub-problem,
and the threshold value $\CL$.
Our implementation requires that the cube of each sub-problem be encoded in the name of the corresponding refutation file.
For example, the refutation file corresponding to refutation $\pi_{00}$ in Figure \ref{fig:decision} is named $\ell_1\_\ell_2.\mathit{proof}$.
The output is a single file containing a refutation of the original instance.
As noted in section \ref{sec:methodology}, stitching operations at the same depth of the decision tree are independent and their combined refutations can be optimized in parallel. Our tool supports this.
Setting the parameter $\CL = 0$ enables optimization after every stitching operation and $\CL = -1$ turns off optimization (only stitching is performed). We denote refutations combined with $\CL=0$ as "fully optimized" and refutations combined with $\CL=-1$ as "unoptimized".

\section{Experiments}
\label{sec:results}
%
To evaluate \stitch, we run it on six benchmarks from the parallel track of last year's SAT competition~\cite{sat2021}.
The chosen benchmarks can be solved by \solver{Paracooba}~\cite{heisinger2020parac} within 1 minute of run-time.
We also attempted running the tool on harder instances from the parallel track. 
While unoptimized proofs can be produced quickly (within a few minutes) on those instances, proof-checking and optimization are both computationally prohibitive due to the limitation of the underlying proof-checker (e.g., \emph{drat-trim} fails  to validate the combined refutations on harder instances even with a 24 hour time limit).
For large refutations, the proof-checker faces memory and run-time bottlenecks on almost all the intermediate optimization steps.
Therefore, we do not consider harder instances in our evaluation, but note that the proposed techniques in principle apply to larger instances once the scalability of the underlying proof-checker improves. 



In our experiments, we compare the checking time and size of unoptimized refutations against fully optimized refutations to show the benefit of optimization. 
We also report the tool run-time to demonstrate that \stitch does not introduce unacceptable overheads.
Finally, we analyze the average checking time and tool run-time for $\CL = 10$, a value empirically determined to perform well.
We perform our evaluation on an Intel Xeon E5-2640 v3 machine with 128 GBytes of DRAM and 16 cores.

Table \ref{Tab:expts} shows the time required for \emph{drat-trim} to check the final refutations for the benchmarks ($T_c$), tool execution time to  combine refutations ($T_g$), and the size of the combined refutations ($S_g$). The time required to check refutations reduces by between $(2.7-7)\times$ for all the benchmarks when full optimization is performed.
Full optimization also results in smaller refutation file sizes, but increases the tool run-time.

\begin{table}[t]
\centering
\begin{footnotesize}
\caption{%
Refutation checking time ($T_c$) (s),
tool run-time ($T_g$) (s),
and size of refutation file ($S_g$) (MB)
for six benchmarks from last year's SAT competition~\cite{sat2021} \label{Tab:expts}}
\setlength{\tabcolsep}{4pt}
\begin{tabular}{lcccccc}
\hline
\multirow{2}{*}{Benchmarks} & \multicolumn{3}{c}{Un-optimized} & \multicolumn{3}{c}{Fully Optimized} \\
                            & 
                            $T_c$(s)  & $T_g$(s)
                         & $S_g$(MB) & $T_c$(s)  & $T_g$(s) & $S_g$(MB) \\
\cmidrule(lr){1-1} \cmidrule(lr){2-4} \cmidrule(lr){5-7}
p01\_lb\_05                 & 987       & 271       & 1700    & 141       & 686      & 184    \\
ktf\_TF-4.tf\_2\_0.02\_18
& 212       & 78        & 385     & 76        & 600      & 77     \\
satch2ways12u               & 1370      & 275       & 1600    & 272       & 836      & 655    \\
pb\_300\_10\_lb\_06         & 163       & 107       & 536     & 36        & 459      & 27     \\
mp1-Nb6T06                  & 241       & 106       & 586     & 44        & 201      & 222    \\
E02F17                      & 417       & 223       & 1500    & 112       & 467      & 294    \\
\hline
\end{tabular}
\end{footnotesize}
\end{table}

Figure \ref{fig:result} compares the average run-time to combine refutations (denoted ``merging'' time) and the average run-time to check refutations for unoptimized, $\CL=10$, and fully optimized refutations.
%
Interestingly, running our tool with $\CL=10$ \emph{decreases} the total validation time (merging + checking) compared to the unoptimized case. 
This points to the benefit of optimizing refutations in parallel---the overhead associated with optimizing refutations can be amortized by the savings in refutation checking time.
Another important observation is that setting $CL_{avg}=10$ reduces the time required to combine refutations compared to the unoptimized case.
We believe the reason is as follows: optimizing refutations decreases their size. When $\CL=10$, we optimize all intermediate refutations with average clause length greater than 10.  Since the intermediate refutations are now smaller, the next stitching operation on this refutation takes lesser time. The time spent in optimizing refutations is mitigated by the savings in stitching time. 


\begin{figure}[t]
    \centering
    \includegraphics[width=0.9\columnwidth]{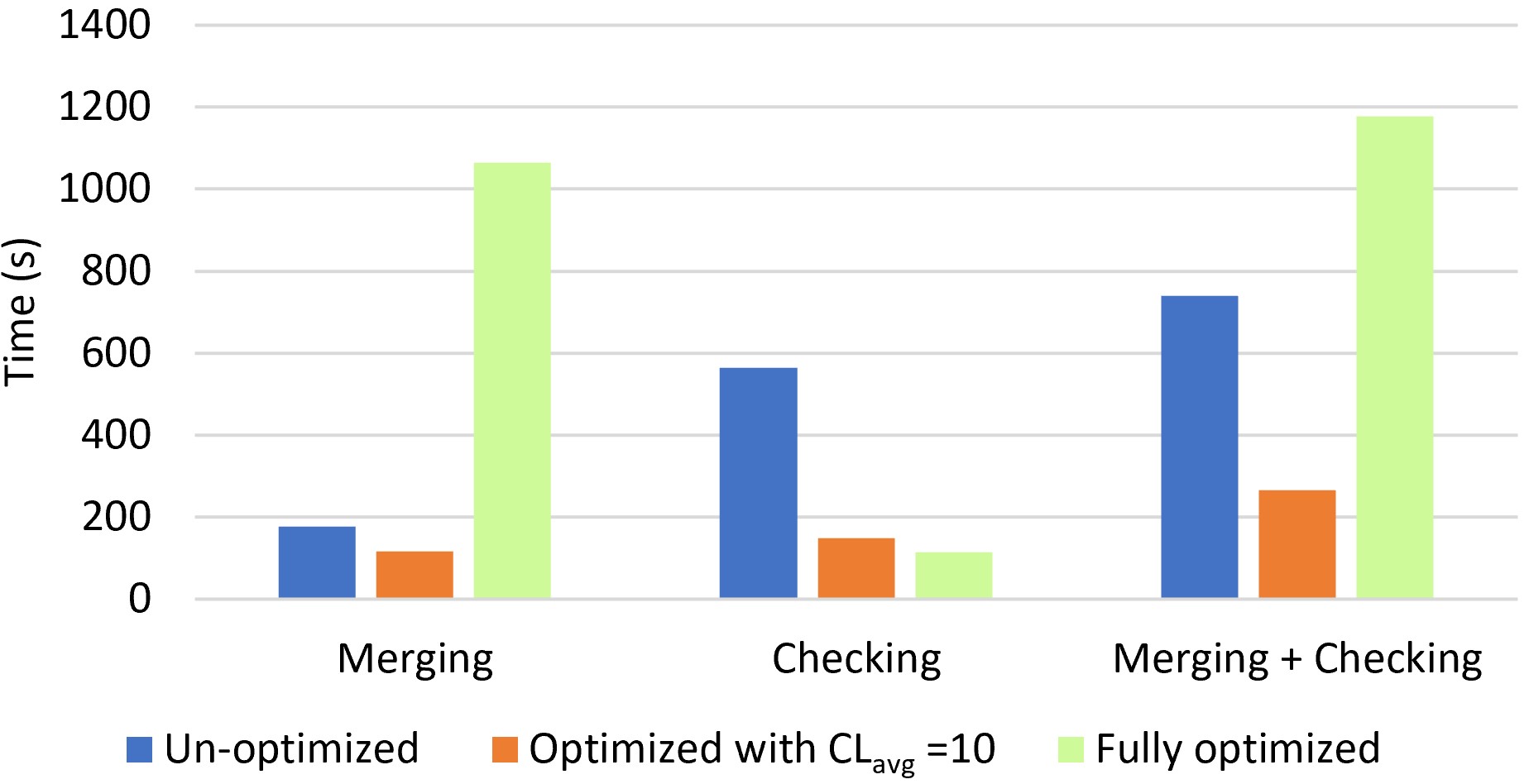}
    \caption{Average merging time and refutation checking time when the refutations are not optimized, optimized with $\CL=10$ and fully optimized }
    \label{fig:result}
\end{figure}

\section{Conclusion}
\label{sec:conc}
We have presented \stitch, a technique that complements Divide-and-Conquer SAT solvers by combining sub-problem refutations into a single refutation for the original instance.
\stitch also uses existing proof-trimming tools to optimize the combined refutation.

\textit{Future Work:} \stitch's run-time overhead can be reduced by performing more stitching operations in parallel.
Currently, only stitching operations at the same tree depth are parallelized, while in principle, any two independent stitching operations
could be parallelized.
Another potential future direction would be to incorporate parallelism in the refutation checker itself, likely requiring extension of the DRAT format to incorporate structural information of the search tree. 
Finally, it would be 
interesting to evaluate alternative measures for guiding the optimization process, such as Literal Block Distance~\cite{lbd}, and to look into additional ways to reduce refutation sizes.

\textit{Acknowledgement:}
This work began as a course project
for Caroline Trippel's CS357S (Fall 2021) at Stanford University.

\clearpage
\bibliographystyle{IEEEtran.bst}
\bibliography{references}

\end{document}